\documentclass[11pt]{article}
\usepackage[margin=30mm]{geometry}
\usepackage{cite}
\usepackage{amsmath,amssymb,amsfonts,amsthm}
\usepackage{algorithmic}
\usepackage{graphicx}
\usepackage{textcomp}
\usepackage{xcolor}
\usepackage{soul,dsfont} 
\usepackage{tikz} 
\usepackage{color}
\usepackage{algorithm} 
\usepackage{array}
\usepackage{eqparbox}
\usepackage{url}
\usepackage{relsize}
\usepackage{stfloats}
\usepackage{graphics}
\usepackage[mathscr]{euscript}
\graphicspath{{figures/}{./}}

\providecommand{\keywords}[1]
{
  \small	
  \textbf{\textit{Keywords---}} #1
}

\newcommand{\blue}{\color{blue}}

\newtheorem{proposition}{Proposition}

\newtheorem{definition}{Definition}
\newtheorem{corollary}{Corollary}
\newtheorem{remark}{Remark}

\title{\bf
Data-driven Analysis of T-Product-based Dynamical Systems}

\author{Xin Mao\thanks{Xin Mao is with the School of Data Science and Society, University of North Carolina at Chapel Hill, Chapel Hill, NC 27599, USA. {\tt\small xinm@unc.edu}} 
\and Anqi Dong\thanks{Anqi Dong is with the Division of Decision and Control Systems and Department of Mathematics, KTH Royal Institute of Technology, SE-100 44 Stockholm, Sweden. {\tt\small anqid@kth.se}} 
\and Ziqin He\footnotemark[3] \and Yidan Mei \thanks{Ziqin He and Yidan Mei are with the Department of Mathematics, University of North Carolina at Chapel Hill, Chapel Hill, NC 27599, USA. {\tt\small zhe21@unc.edu, ymei@unc.edu}} \and Can Chen\thanks{Can Chen is with the School of Data Science and Society and the Department of Mathematics, University of North Carolina at Chapel Hill, Chapel Hill, NC 27599, USA.{\tt\small canc@unc.edu}
}}

\begin{document}

\maketitle
\thispagestyle{empty}
\pagestyle{empty}

\begin{abstract}
A wide variety of data can be represented using third-order tensors, spanning applications in chemometrics, psychometrics, and image processing. However, traditional data-driven frameworks are not naturally equipped to process tensors without first unfolding or flattening the data, which can result in a loss of crucial higher-order structural information. In this article, we introduce a novel framework for the data-driven analysis of T-product-based dynamical systems (TPDSs), where the system evolution is governed by the T-product between a third-order dynamic tensor and a third-order state tensor. In particular, we examine the data informativity of TPDSs concerning system identification, stability, controllability, and stabilizability and illustrate significant computational improvements over traditional approaches by leveraging the unique properties of the T-product. The effectiveness of our framework is demonstrated through numerical examples.
\end{abstract}


\section{Introduction}
Numerous real-world systems, including those found in image processing, biological systems, and social sciences, exhibit complex, multi-dimensional relationships, where the states are often represented as third-order or higher-order tensors \cite{ding2018tensor,chen2015functional,williams2018unsupervised,khoromskij2018tensor}. Multilinear dynamical systems, newly proposed in recent years, extend classical linear systems theory that offer a powerful framework for modeling these tensor-based systems that cannot be adequately captured by traditional methods \cite{rogers2013multilinear,chen2021multilinear}. One of the most effective tools for working with multilinear dynamical systems is the T-product, a mathematical operation that extends matrix multiplications to third-order tensors in a manner analogous to matrix operations such as eigenvalue decomposition and singular value decomposition \cite{braman2010third,kilmer2013third}. 

The T-product framework provides a versatile approach to analyzing and controlling multilinear dynamical systems in a variety of fields, including physics \cite{chang2022t,lund2020tensor}, engineering \cite{tarzanagh2018fast}, and
biology \cite{yang2020sparse,kilmer2013third}. Specifically, the T-product empowers researchers to perform sophisticated operations on multidimensional data, e.g., images, which are often represented as third-order tensors. Traditional matrix-based methods struggle to capture the complex relationships between different dimensions of images including height, width, and color channels. In contrast,  T-product-based methods preserve the inherent multidimensional structure and have potential applications in tasks such as image denoising, image compression, image segmentation, and feature extraction, as demonstrated by recent research \cite{kilmer2011factorization, zhang2016exact, yang2020sparse, ahmadi2023fast, gilman2022grassmannian}.  

T-product-based dynamical systems (TPDSs) are systems whose evolution is governed by the T-product between a third-order dynamic tensor and a third-order state tensor. The concept was first proposed by Hoover et al. \cite{hoover2021new} as a generalization of linear time-invariant (LTI) systems. TPDSs offer a powerful framework for capturing complex interactions in three-dimensional data. The development of tensor decomposition techniques and circulant algebra has enabled a seamless extension of linear systems theory to TPDSs, covering fundamental concepts like explicit solutions, stability, controllability, and observability. Nevertheless, the lack of computational tools for their data-driven analysis has limited their use in practical applications. This gap is particularly evident in areas requiring stability, controllability, and system identification from observational data. 

Data-driven analysis and control have garnered significant attention in recent years \cite{van2020data,markovsky2021behavioral,berberich2020trajectory}. The origins of this field can be traced back to the early 1980s, with the pioneering work on the fundamental lemma (also known as persistency of excitation)  \cite{willems2005note}, which laid the theoretical foundation for using input-output data to infer system properties. Recently, Van Waarde et al. \cite{van2020data} presented a novel data-driven analysis and control framework to investigate the data informativity of linear time-invariant
systems, where data are not informative enough to uniquely identify the system (i.e., the fundamental lemma fails).
Although data-driven research has attracted considerable interest, the exploration of data-driven approaches specifically for tensor-based systems remains relatively underdeveloped. 

We are dedicated to developing a data-driven approach for TPDSs. Extending traditional system identification, stability, and controllability analyses to third-order tensor-based systems is challenging due to the curse of dimensionality. We use the properties of the T-product to address this challenge. Our work provides effective and efficient conditions of data informativity for system identification, stability, controllability, and stabilizability of TPDSs. Additionally, we show how T-product-based conditions offer significant computational advantages over existing unfolding-based approaches, demonstrating their applicability through numerical experiments. The rest of the article is structured as follows. In Section \ref{sec:prelim}, we begin with an overview of T-product operations. In Section \ref{sec:data-driven}, we introduce  TPDSs and examine the data informativity of system identification, stability, controllability, and stabilizability. In Section \ref{sec:num}, we provide numerical examples for the proposed methods. Finally, we conclude our work with a discussion of future directions in Section \ref{sec:conclusion}.

\section{Preliminaries}\label{sec:prelim}
Tensors can be considered as multidimensional arrays, which extend the concepts of vectors and matrices to higher-dimensional settings \cite{chen2024tensor, kolda2009tensor, ragnarsson2012block}. The order of a tensor is defined as the number of dimensions. Of particular interest in this article are third-order tensors, which we denote by $\mathscr{T}\in\mathbb{R}^{n\times m\times s}$. We first introduce the notion of the T-product, an effective operation for manipulating third-order tensors, that enables their multiplication through the notion of circular convolution \cite{kilmer2013third, kilmer2011factorization, zhang2016exact}.

\begin{definition}[\bf \textit{T-product}]
The T-product between two third-order tensor $\mathscr{T}\in\mathbb{R}^{n\times m\times s}$ and $\mathscr{S}\in\mathbb{R}^{m\times r\times s}$, denoted by $\mathscr{T}\star\mathscr{S}$, is defined as 
\begin{equation}
\mathscr{T}\star\mathscr{S} = \texttt{fold}\Big(\texttt{bcirc}(\mathscr{T})\texttt{unfold}(\mathscr{S})\Big) \in\mathbb{R}^{n\times r\times s},
\end{equation}
where  $\texttt{bcirc}(\cdot)$ and $\texttt{unfold}(\cdot)$ are defined as
\begin{align*}\label{eq:bcirc}
\texttt{bcirc}(\mathscr{T})&= \begin{bmatrix}
        \mathscr{T}_{::1} & \mathscr{T}_{::s} & \cdots & \mathscr{T}_{::2}\\
        \mathscr{T}_{::2} & \mathscr{T}_{::1} & \cdots & \mathscr{T}_{::3}\\
        \vdots & \vdots & \ddots & \vdots\\
        \mathscr{T}_{::s} & \mathscr{T}_{::(s-1)} & \cdots & \mathscr{T}_{::1}
\end{bmatrix}\in\mathbb{R}^{ns\times ms},\\
\texttt{unfold}(\mathscr{S}) &= \begin{bmatrix}
        \mathscr{S}_{::1} &
        \mathscr{S}_{::2}&
        \cdots&
        \mathscr{S}_{::s}
\end{bmatrix}^\top\in\mathbb{R}^{ms\times r},
\end{align*}
and  \texttt{fold} is the  reverse operation of \texttt{unfold}.
\end{definition}

It is noteworthy that many fundamental matrix operations, including identity, diagonal, transpose, inverse, and orthogonality can be also generalized to third-order tensors using the T-product:
\begin{enumerate}
    \item {\bf T-identity:} T-identity tensor $\mathscr{I}$ is defined as having the first frontal slice (i.e., $\mathscr{I}_{::1}$) as the identity matrix with all other frontal slices consisting of zeros.
    \item {\bf T-diagonal:} T-diagonal tensor is defined such that each frontal slice is a diagonal matrix.
    \item {\bf T-transpose:} T-transpose of $\mathscr{T}\in\mathbb{R}^{n\times m\times s}$ is obtained by transposing each of the frontal slices and then reversing the order of the transposed frontal slices from 2 to $s$.
    \item {\bf T-inverse:} T-inverse of $\mathscr{T}\in\mathbb{R}^{n\times n\times s}$, denoted by $\mathscr{T}^{-1}$, is defined as $\mathscr{T}\star\mathscr{T}^{-1}=\mathscr{T}^{-1}\star\mathscr{T}=\mathscr{I}$ (similarly for left and right T-inverse).
    \item {\bf T-orthogonal:} A third-order tensor $\mathscr{T}\in\mathbb{R}^{n\times n\times s}$ is called T-orthogonal is $\mathscr{T}\star\mathscr{T}^\top = \mathscr{T}^\top\star\mathscr{T}=\mathscr{I}$.
\end{enumerate}

\noindent As a matter of fact, the operations above can be computed through the circulant operation \texttt{bcirc}. For example, the T-inverse can be attained as $$
\mathscr{T}^{-1}:=\texttt{un-bcirc}(\texttt{bcirc}(\mathscr{T})^{-1}),
$$ where \texttt{un-bcirc} denotes the reverse operation of \texttt{bcirc}. With a slight abuse of notation, we use the same superscript for both matrix and T-product-based operations, e.g., matrix transpose and T-transpose are denoted by $(\cdot)^{\top}$.

Notably, eigenvalue decomposition and singular value decomposition can be defined for third-order tensors in a similar manner as matrices through the T-product \cite{kilmer2013third,braman2010third}.

\begin{definition}[\bf \textit{T-eigenvalue decomposition}] The
T-eigen-value decomposition of a third-order tensor $\mathscr{T}\in\mathbb{R}^{n\times n\times s}$ is defined as
\begin{equation}
    \mathscr{T}=\mathscr{U}\star\mathscr{D}\star\mathscr{U}^{-1},
\end{equation}
where $\mathscr{U}\in\mathbb{R}^{n\times n\times s}$ and $\mathscr{D}\in\mathbb{R}^{n\times n\times s}$ is T-diagonal such that $\mathscr{D}_{jj:}\in\mathbb{R}^{s}$ are referred to the eigentuples of $\mathscr{T}$.
\end{definition}

The T-eigenvalue decomposition can be computed using circulant operation \(\texttt{bcirc}\) and matrix eigenvalue decomposition. However, employing the discrete Fourier transform can significantly expedite the process. In particular, a circulant matrix can be block diagonalized via left and right multiplication by a block diagonal discrete Fourier transform matrix. The Fourier transform $\mathcal{F}\{\texttt{bcirc}(\mathscr{T})\}$  is defined as
\begin{align*}
\mathcal{F}\{\texttt{bcirc}(\mathscr{T})\} =& (\textbf{F}_n\otimes \textbf{I}) \texttt{bcirc}(\mathscr{T}) (\textbf{F}_n\otimes \textbf{I})^\top \\
=& \begin{bmatrix}
       \textbf{T}_1 & \textbf{0} &\cdots & \textbf{0}\\
        \textbf{0} & \textbf{T}_2  &\cdots & \textbf{0}\\
        \vdots & \vdots & \ddots & \vdots\\
        \textbf{0} & \textbf{0} & \cdots & \textbf{T}_s
   \end{bmatrix},
\end{align*}
where $\textbf{F}_n\in\mathbb{R}^{n\times n}$ is the discrete Fourier transform matrix defined as
\begin{equation*}
    \textbf{F}_n= \frac{1}{\sqrt{n}}\begin{bmatrix}
        1 & 1 & 1 &\cdots & 1\\
        1 & \omega & \omega^2 & \cdots & \omega^{n-1}\\
        \vdots & \vdots & \vdots & \ddots & \vdots\\
        1 & \omega^{n-1} & \omega^{2(n-1)} & \cdots & \omega^{(n-1)^2}
    \end{bmatrix},
\end{equation*}
with $\omega = \exp{\{\frac{-2\pi i}{n}\}}$ (note that $i$ denotes the imaginary number here), and $\otimes$ denotes the Kronecker product. Hence, the T-eigenvalue decomposition of $\mathscr{T}$ can be constructed through the eigenvalue decompositions of $\textbf{T}_j$.  Given that $\textbf{T}_j = \textbf{U}_j\textbf{D}_j\textbf{U}^{-1}_j$, $\mathscr{U}$ can be recovered by
\begin{equation*}
    \texttt{bcric}(\mathscr{U}) = (\textbf{F}_n\otimes \textbf{I})^{*}\texttt{blkdaig}(\textbf{U}_1,\dots,\textbf{U}_s) (\textbf{F}_n\otimes \textbf{I}),
\end{equation*}
where operation $\texttt{blkdiag}$ denotes the MATLAB block diagonal function, and the superscript $[\cdot]^\star$ denotes the conjugate transpose. The T-diagonal tensor $\mathscr{D}$ can be obtained in a similar manner. 

 
\begin{definition}[\bf \textit{T-singular value decomposition}]
T-singul-ar value decomposition (T-SVD) of  a third-order tensor $\mathscr{T}\in\mathbb{R}^{n\times m\times s}$ is defined as
    \begin{equation}
        \mathscr{T}=\mathscr{U}\star\mathscr{S}\star\mathscr{V}^{\top},
    \end{equation}
where $\mathscr{U}\in\mathbb{R}^{n\times n\times s}$ and $\mathscr{V}\in\mathbb{R}^{m\times m\times s}$ are T-orthogonal, and $\mathscr{S}\in\mathbb{R}^{n\times m\times s}$ is a T-(rectangle) diagonal tensor such that 
$\mathscr{S}_{jj:}\in\mathbb{R}^{s}$ are referred to the singular tuples of $\mathscr{T}$.
\end{definition}

The T-singular value decomposition can be computed analogously by applying the Fourier transform \(\mathcal{F}\{\texttt{bcirc}(\mathscr{T})\}\) and then performing singular value decomposition on the block diagonal matrices \(\textbf{T}_j\). We will show that both T-eigenvalue decomposition and T-singular value decomposition play critical roles in the data-driven analysis of TPDSs.

\section{Data-driven Analysis of TPDSs}\label{sec:data-driven}
We are now positioned to conduct data-driven analysis of TPDSs, which are generally expressed in the form of
\begin{equation}\label{eq:tpds}
\mathscr{X}(t+1) = \mathscr{A}\star\mathscr{X}(t),
\end{equation}
with $\mathscr{A}\in\mathbb{R}^{n\times n\times r}$ represents the state transition tensor, and $\mathscr{X}(t)\in\mathbb{R}^{n\times h\times r}$ denotes the state. Significantly, the TPDS  \eqref{eq:tpds} can be transformed into LTI systems using  \texttt{bcirc} and \texttt{unfold}, resulting in two linear representations, i.e., 
\begin{align*}
    \texttt{unfold}\left(\mathscr{X}(t+1)\right) &= \texttt{bcirc}(\mathscr{A})\texttt{unfold}(\mathscr{X}(t)),\\
    \texttt{bcirc}(\mathscr{X}(t+1)) &= \texttt{bcirc}(\mathscr{A})\texttt{bcirc}(\mathscr{X}(t)).
\end{align*}
As a result, data-driven analysis techniques from LTI systems can be effectively extended to TPDSs.

Assume the state  data tensors by concatenating are collected at the second mode, i.e., 
\begin{align*}
    \mathscr{X}_0 &= 
    \begin{bmatrix}
        \mathscr{X}(0) & \mathscr{X}(1) & \cdots & \mathscr{X}(l-1)
    \end{bmatrix}\in\mathbb{R}^{n\times lh\times r},\\
    \mathscr{X}_1 &= 
    \begin{bmatrix}
        \mathscr{X}(1) & \mathscr{X}(2) & \cdots & \mathscr{X}(l)
    \end{bmatrix}\in\mathbb{R}^{n\times lh\times r}.
\end{align*}
While the unfolded linear representations are useful for studying the data-driven analysis of TPDSs, the full computational benefits are achieved by leveraging the properties of the T-product. In the following, we first examine the data informativity of TPDSs with respect to system identification, stability, controllability, and stabilizability. The efficiency of T-product-based computations in expressing these conditions is also illustrated, along with numerical examples.

\subsection{System Identification}\label{subsec:info}
The data informativity for system identification of TPDSs involves determining the conditions under which the state transition tensor $\mathscr{A}$ can be uniquely identified.

\begin{definition}[\bf \textit{System identification}]
    We say the data $(\mathscr{X}_0,\mathscr{X}_1)$ are informative for system identification if the state transition tensor $\mathscr{A}$ can be uniquely identified. 
\end{definition}


\begin{proposition}\label{prop:si}
    The data $(\mathscr{X}_0,\mathscr{X}_1)$ are informative for system identification if and only if the rank of  $\texttt{bcirc}(\mathscr{X}_0)$ is equal to $nr$.
\end{proposition}
 
\begin{proof}
    Substituting data $(\mathscr{X}_0,\mathscr{X}_1)$ into the TPDS \eqref{eq:tpds} gives $\mathscr{X}_1=\mathscr{A}\star\mathscr{X}_0$. Due to the properties of block circulant matrix operation, it follows that 
    \begin{align*}
    \texttt{bcirc}(\mathscr{A})&=\texttt{bcirc}(\mathscr{X}_1)\texttt{bcirc}(\mathscr{X}_0)^\dagger.
    \end{align*}
 According to linear matrix theory, $\texttt{bcirc}(\mathscr{A})$ can be uniquely determined if and only if  $\texttt{bcirc}(\mathscr{X}_0)$ has full row rank, i.e., $nr$. The result thus follows immediately. 
\end{proof}

We choose to use the second linear representation because the rank of $\texttt{bcirc}(\mathscr{X}_0)$ can be determined through T-product-based computations. Specifically, we can leverage the T-SVD of $\mathscr{X}_0$ and the associated block diagonal matrices in the Fourier domain.

\begin{corollary}\label{coro:sisvd}
    The data $(\mathscr{X}_0,\mathscr{X}_1)$ are informative for system identification if and only if the singular tuples of $\mathscr{X}_0$ in the Fourier domain contain non-zero entries. 
\end{corollary}
 
\begin{proof}
    Based on the finding in \cite{lund2020tensor}, the singular values of $\texttt{bcirc}(\mathscr{X}_0)$ are the union of elements from the singular tuples of $\mathscr{X}_0$ in the Fourier domain. According to linear matrix theory, the rank of $\texttt{bcirc}(\mathscr{X}_0)$ is determined by the number of its non-zero singular values. Therefore, the result follows from Proposition \ref{prop:si}.
\end{proof}

\begin{corollary}\label{coro:sirk}
    The data $(\mathscr{X}_0,\mathscr{X}_1)$ are informative for system identification if and only if
    the sum of the ranks of the block diagonal matrices of $\mathcal{F}
    \{\texttt{bcirc}(\mathscr{X}_0)\}$ is equal to $nr$.
\end{corollary}
 
\begin{proof}
     The singular tuples of $\mathscr{X}_0$ in the Fourier domain can be computed from the SVDs of the  block diagonal matrices of $\mathcal{F}
    \{\texttt{bcirc}(\mathscr{X}_0)\}$. Therefore,  each block diagonal matrix has full rank if and only if the corresponding singular tuples in the Fourier domain contains non-zero entries. The result then follows from Proposition \ref{prop:si}. 
\end{proof}

\begin{remark}\label{rk:1}
    The time complexity of directly computing the rank of $\texttt{bcirc}(\mathscr{X}_0)$ is about $\mathcal{O}(n^2r^3lh)$ (assuming $n<lh$). Note that the complexity reduces to $\mathcal{O}(n^2r^2lh)$ by using the first linear representation. On the other hand, both Corollaries \ref{coro:sisvd} and \ref{coro:sirk} only require  $\mathcal{O}(n^2rlh)$ operations for determining the data informativity for system identification of TPDSs. Therefore, T-product-based computations offer computational benefits over the unfolding-based approach.
\end{remark}

\subsection{Stability}
The data informativity for the stability of TPDSs involves determining whether any state transition tensor $\mathscr{A}$ identified from the data is stable. Specifically, a state transition tensor $\mathscr{A}$ is considered stable if $\texttt{bcirc}(\mathscr{A})$ is stable, meaning that the eigenvalues of $\texttt{bcirc}(\mathscr{A})$ lie within the unit circle.

\begin{definition}[\bf \textit{Stability}]
    We say the data $(\mathscr{X}_0,\mathscr{X}_1)$ are informative for stability if any state transition tensor $\mathscr{A}$ identified from the data is stable.
\end{definition}


\begin{proposition}\label{prop:stability}
The data $(\mathscr{X}_0,\mathscr{X}_1)$ are informative for stability if the following conditions are satisfied.
\begin{enumerate}
        \item[(i)] the rank of $\texttt{bcirc}(\mathscr{X}_0)$ is equal to $nr$;
        \item[(ii)] $\texttt{bcirc}(\mathscr{X}_1\star \mathscr{X}_0^\dagger)$ is stable (i.e., its eigenvalues are less than or equal to one) for any right T-inverse $\mathscr{X}_0^\dagger$,
\end{enumerate}
\end{proposition}
 
\begin{proof}
Based on the finding of the data informativity for stability of LTI systems \cite{van2020data},  the matrix data $(\texttt{bcirc}(\mathscr{X}_0), \texttt{bcirc}(\mathscr{X}_1))$ are informative for stability if and only if the rank of $\texttt{bcirc}(\mathscr{X}_0)$ is equal to $nr$ and $\texttt{bcirc}(\mathscr{X}_1) \texttt{bcirc}(\mathscr{X}_0)^\dagger$ is stable. Additionally, according to the properties of block circulant matrices, it follows that
\begin{equation*}
    \texttt{bcirc}(\mathscr{X}_1) \texttt{bcirc}(\mathscr{X}_0)^\dagger = \texttt{bcirc}(\mathscr{X}_1\star \mathscr{X}_0^\dagger).
\end{equation*}
Therefore, the result follows immediately. 
\end{proof}

Similar to LTI systems, the data $(\mathscr{X}_0,\mathscr{X}_1)$ are informative for stability only if the system can be uniquely identified (i.e., the data are informative for system identification). In the following, we exploit the T-eigenvalue decomposition/T-SVD and the corresponding block diagonal matrices in the Fourier domain to articulate the aforementioned conditions.

\begin{corollary}\label{coro:stability1}
    The data $(\mathscr{X}_0,\mathscr{X}_1)$ are informative for stability if and only if the following conditions:
\begin{enumerate}
    \item[(i)]  the singular tuples of $\mathscr{X}_0$ in the Fourier domain contain non-zero entries; 
    \item [(ii)] the eigentuples of $\mathscr{X}_1\star \mathscr{X}_0^\dagger$ in the Fourier domain contain entries that are less than or equal to one for any right T-inverse $\mathscr{X}_0^\dagger$,
\end{enumerate}
 are satisfied. 
\end{corollary}
 
\begin{proof}
    The first condition follows Corollary \ref{coro:sisvd}. For the second condition, based on the finding in \cite{lund2020tensor}, the eigenvalues of $\texttt{bcirc}(\mathscr{X}_1\star \mathscr{X}_0^\dagger)$ are the union of elements from the eigentuples of $\mathscr{X}_1\star\mathscr{X}_0^\dagger$ in the Fourier domain. Therefore, the result follows from Proposition \ref{prop:stability}.
\end{proof}

\begin{corollary}\label{coro:stability2}
    The data $(\mathscr{X}_0,\mathscr{X}_1)$ are informative for stability if and only if the following conditions:
    \begin{enumerate}
    \item[(i)] the sum of the ranks of the block diagonal matrices of $\mathcal{F}
    \{\texttt{bcirc}(\mathscr{X}_0)\}$ is equal to $nr$;
        \item[(ii)] the block diagonal matrices of $\mathcal{F}\{\texttt{bcirc}(\mathscr{X}_1\star \mathscr{X}_0^\dagger)\}$ are stable for anyright T-inverse $\mathscr{X}_0^\dagger$.
    \end{enumerate}
    are satisfied.
\end{corollary}
 
\begin{proof}
The first condition follows Corollary \ref{coro:sirk}. For the second condition, the eigentuples of $\mathscr{X}_1\star \mathscr{X}_0^\dagger$ in the Fourier domain can be computed from the eigenvalue decomposition of the block diagonal matrices of  $\mathcal{F}
    \{\texttt{bcirc}(\mathscr{X}_1\star \mathscr{X}_0^\dagger)\}$. Therefore, each block diagonal matrix is stable if and only if the corresponding eigentuple in the Fourier domain contains non-zero entries. The result follows from Proposition \ref{prop:stability}.
\end{proof}

\begin{remark}
    The time complexity of directly computing the eigenvalues of $\texttt{bcirc}(\mathscr{X}_1\star \mathscr{X}_0^\dagger)$ is estimated as $\mathcal{O}(n^3r^3)$. On the contrary, both Corollaries \ref{coro:stability1} and \ref{coro:stability2} only involve $\mathcal{O}(n^3r)$ operations for determining the data informativity for stability of TPDSs. Hence, T-product-based computations are advantageous compared to the unfolding-based approach.
\end{remark}

\subsection{Controllability \& Stabilizability}
The data informativity for controllability/stabilizability of TPDSs entails determining the conditions under which any system identified from the data is controllable/stabilizable. First, we introduce the model of TPDSs with control which is defined as
\begin{equation}\label{eq:tpdsc}
    \mathscr{X}(t+1) = \mathscr{A}\star\mathscr{X}(t)+\mathscr{B}
    \star\mathscr{U}(t),
\end{equation}
where $\mathscr{B}\in\mathbb{R}^{n\times m\times r}$ represents the control matrix, and $\mathscr{U}(t)\in\mathbb{R}^{m\times h
\times r}$ denotes the control input. The system \eqref{eq:tpdsc} is said to be controllable if for any initial state $\mathscr{X}(0)$ and target state $\mathscr{X}(T)$, there exists a sequence of inputs $\mathscr{U}(t)$ that drives the system from $\mathscr{X}(0)$ to $\mathscr{X}(T)$ \cite{chen2024t}. The system \eqref{eq:tpdsc} is considered stabilizable if there exists a sequence of inputs of the form 
$$
\mathscr{U}(t)=\mathscr{K}\star\mathscr{X}(t),
$$ 
for $\mathscr{K}\in\mathbb{R}^{m\times n\times r}$, such that the new system $\mathscr{A}+\mathscr{B}\star\mathscr{K}$ is stable. Finally, suppose that the input data is collected as
\begin{equation*}
    \mathscr{U}_0 = \begin{bmatrix}
        \mathscr{U}(0) & \mathscr{U}(1) & \cdots & \mathscr{U}(l-1)
    \end{bmatrix}\in\mathbb{R}^{m\times lh\times r}.
\end{equation*}

The data informativity of controllability and stabilizability for TPDSs can be defined as follows. 
 
\begin{definition}[\bf \textit{Controllability}]
    We say data $(\mathscr{U}_0,\mathscr{X}_0,\mathscr{X}_1)$ are informative for controllability if any pair $(\mathscr{A},\mathscr{B})$ identified by the data is controllable. 
\end{definition}
 
\begin{definition}[\bf \textit{Stabilizability}]
    We say data $(\mathscr{U}_0,\mathscr{X}_0,\mathscr{X}_1)$ are informative for stabilizability if any pair $(\mathscr{A},\mathscr{B})$ identified by the data is stabilizable. 
\end{definition}


\begin{proposition}\label{prop:info_control}
    The data $(\mathscr{U}_0,\mathscr{X}_0,\mathscr{X}_1)$ are informative for controllability if and only if the rank of $\texttt{bcirc}(\mathscr{X}_1-\lambda \mathscr{X}_0)$ is equal to $nr$ for any $\lambda\in\mathbb{C}$.
\end{proposition}
 
\begin{proof}
    Based on the finding of the data informativity for controllability of LTI systems \cite{van2020data},  the matrix data $(\texttt{bcirc}(\mathscr{X}_0), \texttt{bcirc}(\mathscr{X}_1))$ are informative for controllability if and only if the rank of $\texttt{bcirc}(\mathscr{X}_1)-\lambda \texttt{bcirc}(\mathscr{X}_0)$ is equal to $nr$ for any $\lambda\in\mathbb{C}$. Moreover, according to the properties of block circulant matrices, it follows that
    \begin{equation*}
       \texttt{bcirc}(\mathscr{X}_1)-\lambda \texttt{bcirc}(\mathscr{X}_0) = \texttt{bcirc}(\mathscr{X}_1-\lambda\mathscr{X}_0),
    \end{equation*}
and the result follows immediately.
\end{proof}

The following two corollaries can be proven similarly as Corollaries \ref{coro:sisvd} and \ref{coro:sirk} with T-SVD and block diagonal matrices in the Fourier domain.

\begin{corollary}\label{coro:con1}
    The data $(\mathscr{X}_0,\mathscr{X}_1)$ are informative for controllability if and only if the singular tuples of $\mathscr{X}_1-\lambda \mathscr{X}_0$ in the Fourier domain
     contain non-zero entries for any $\lambda\in\mathbb{C}$. 
\end{corollary}

\begin{corollary}\label{coro:con2}
    The data $(\mathscr{X}_0,\mathscr{X}_1)$ are informative for stabilizability if and only if the sum of the ranks of the block diagonal matrices of $\mathcal{F}
    \{\texttt{bcirc}(\mathscr{X}_1-\lambda \mathscr{X}_0)\}$ is equal to $nr$ for any $\lambda\in\mathbb{C}$.
\end{corollary}

For the data informativity regarding the stabilizability of TPDSs, an additional condition of $|\lambda| \geq 1$ is required, as established by \cite{van2020data}. Additionally, the computational complexity analysis follows similar principles as those presented in Remark \ref{rk:1}.

\section{Numerical Examples}\label{sec:num}
We  proceed to illustrate our framework with the following numerical experiments. All experiments in this section were conducted on a platform equipped with an M1 Pro CPU and 16GB of memory. The code used for these experiments is available at {\blue \url{https://github.com/dytroshut/TPDSs}}.

\subsection{System Identification}
Identifying the underlying systems from imaging data is critical to extract valuable information, enhance predictive capabilities, and improve overall system performance. In this example, we evaluated our approach for determining the data informativity for system identification of TPDSs by applying Corollary \ref{coro:sirk}. We first randomly generated imaging data  $\mathscr{X}(0)$, $\mathscr{X}(1)$, $\dots$, $\mathscr{X}(l)\in\mathbb{R}^{2\times 2\times r}$ and constructed the state data tensor $\mathscr{X}_0\in \mathbb{R}^{2\times 2l\times r}$ with $l=10$ and $r= 2^{p}$ for $p=2,3,\dots,10$. We then computed the ranks of all block matrices in $\mathcal{F}\{\texttt{bcirc}(\mathscr{X}_0)\}$ and compared the efficiency of our approach with Proposition \ref{prop:si} that directly computes the rank of $\texttt{bcirc}(\mathscr{X}_0)$. The computation time for each rank relative to the corresponding dimension $r$ is shown in Fig.~\ref{fig:rank}a, demonstrating that our approach significantly outperforms the direct rank computation. Fig.\ref{fig:rank}b also shows that our results are consistent with the complexity analysis.


\begin{figure*}[t]
    \centering
    \includegraphics[width=\linewidth]{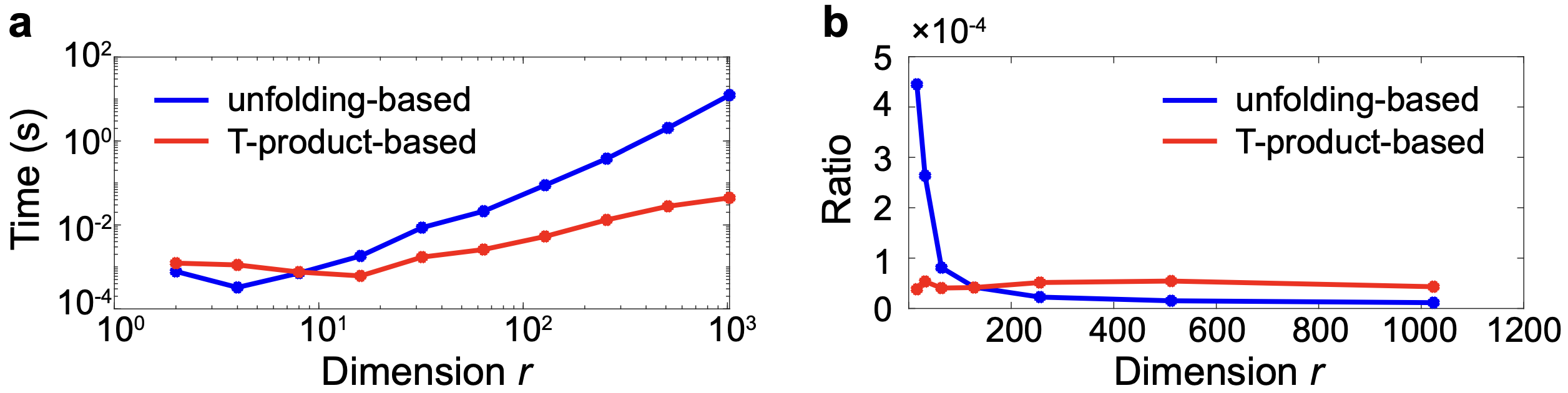}
    \caption{Computational time comparison in determining the data informativity for system identification between the unfolding-based and T-product-based approaches. \textbf{a.} Log-log plot of computational time with respect to the dimension of the third mode $r$. \textbf{b.} Ratio of time to the dimension of the third mode $r$ (i.e., time/$r^3$ for the unfolding-based approach and time/$r$ for the T-product-based approach). }
    \label{fig:rank}
\end{figure*}

\begin{figure*}[t]
    \centering
    \includegraphics[width=\linewidth]{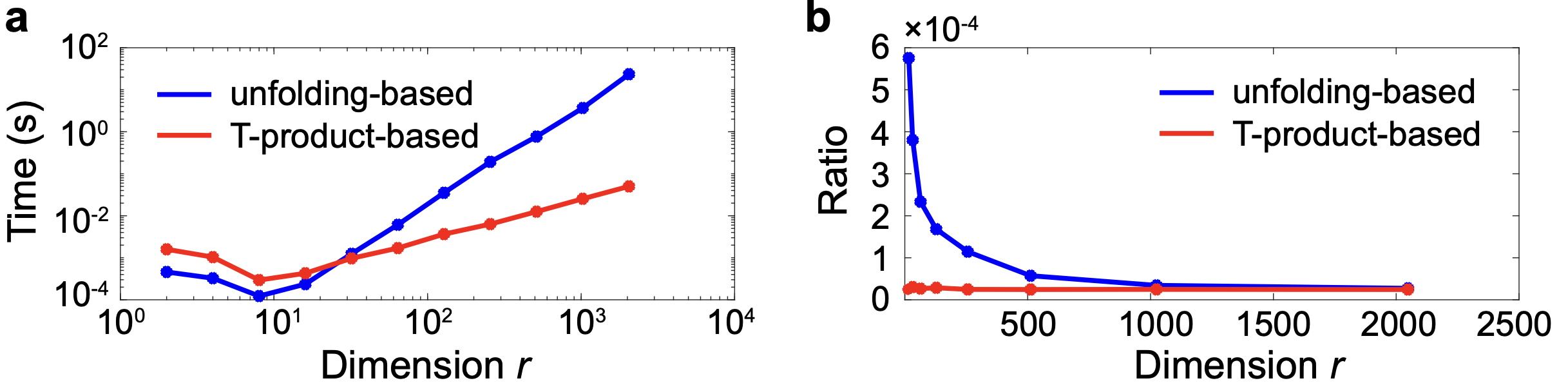}
    \caption{Computational time comparison in determining the data informativity for stability between the unfolding-based and T-product-based approaches. \textbf{a.} Log-log plot of computational time with respect to the dimension of the third mode $r$. \textbf{b.} Ratio of time to the dimension of the third mode $r$ (i.e., time/$r^3$ for the unfolding-based approach and time/$r$ for the T-product-based approach). }
    \label{fig:stable}
\end{figure*}

\subsection{Stability}
Assessing the stability of imaging data is crucial for ensuring that the interpretations and analyses derived from this data are reliable and accurate. In this example, we evaluated our approach for determining the data informativity for stability of TPDSs by applying Corollary \ref{coro:stability2}.  We first randomly generated  imaging data  $\mathscr{X}(0)$, $\mathscr{X}(1)$, $\dots$, $\mathscr{X}(l)\in\mathbb{R}^{2\times 2\times r}$ and constructed the state data tensors $\mathscr{X}_0$, $\mathscr{X}_1\in \mathbb{R}^{2\times 2l\times r}$ with $l=10$ and $r= 2^{p}$ for $p=2,3,\dots,11$. After ensuring the first condition, we computed the eigenvalue decompositions of all block matrices in $\mathcal{F}\{\texttt{bcirc}(\mathscr{X}_1 \star \mathscr{X}_0^\dagger)\}$ and comapred the efficiency of our approach with Proposition \ref{prop:stability} which directly computes the eigenvalue decomposition of $\texttt{bcirc}(\mathscr{X}_1 \star \mathscr{X}_0^\dagger)$. As with the first example, the computational time for using Corollary \ref{coro:stability2} is significantly less than that of the unfolding-based approach as $r$ increases, see Fig. \ref{fig:stable}a. Moreover, this result is consistent with our complexity analysis, see Fig. \ref{fig:stable}b.

\subsection{Controllability}

Determining controllability from imaging data is essential for optimal control design in disease treatment within the medical field. In this example, we evaluated our approach for determining the data informativity for controllability of TPDSs by applying Corollary \ref{coro:con2}. We first randomly generated imaging data  $\mathscr{X}(0)$, $\mathscr{X}(1)$, $\dots$, $\mathscr{X}(l)\in\mathbb{R}^{2\times 2\times r}$ and constructed the state data tensors $\mathscr{X}_0$, $\mathscr{X}_1\in \mathbb{R}^{2\times 2l\times r}$ with $l=10$ and $r= 2^{p}$ for $p=2,3,\dots,9$.  We then computed the ranks of all block matrices in $\mathcal{F}\{\texttt{bcirc}(\mathscr{X}_1-\lambda \mathscr{X}_0)\}$ and compared the efficiency of our approach with Proposition \ref{prop:info_control} that directly computes the rank of $\texttt{bcirc}(\mathscr{X}_1-\lambda \mathscr{X}_0)$. Here, we used the MATLAB symbolic computation to compute the ranks (i.e., symbolic ranks).  As with the first two examples, our approach is significantly faster in determining the data informativity for controllability when $r$ is large compared to the unfolding-based approach, see Fig. \ref{fig:info}a. Again, our results  align with the complexity analysis, see Fig. \ref{fig:info}b.



\begin{figure*}[htb!]
    \centering
    \includegraphics[width=\linewidth]{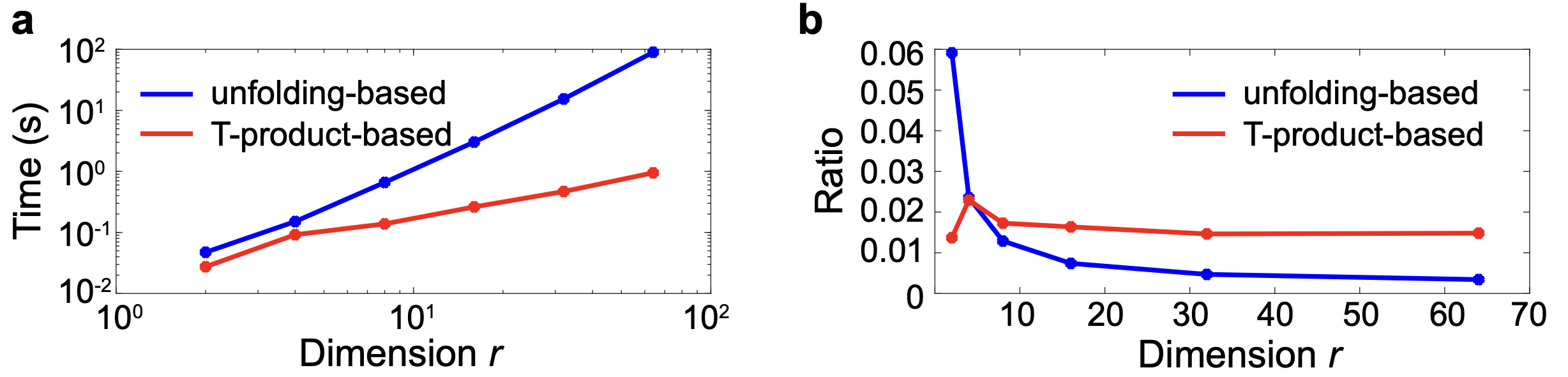}
    \caption{Computational time comparison in determining the data informativity for controllability between the unfolding-based and T-product-based approaches. \textbf{a.} Log-log plot of computational time with respect to the dimension of the third mode $r$. \textbf{b.} Ratio of time to the dimension of the third mode $r$ (i.e., time/$r^3$ for the unfolding-based approach and time/$r$ for the T-product-based approach).}
    \label{fig:info}
\end{figure*}




\section{Conclusion}\label{sec:conclusion}
In this article, we introduced a data-driven analysis framework for TPDSs, where the system evolution is governed by the T-product. We established effective and efficient criteria in determining the data informativity for system identification, stability, controllability, and stabilizability of TPDSs by leveraging the unique properties of the T-product. We further offered detailed complexity analyses for the proposed criteria and verified them with numerical examples.  In the future, it will be valuable to explore the data-driven control of TPDSs, e.g., data informativity for state feedback and quadratic regulators. Additionally, applying the data-driven framework of TPDSs to real-world tasks, including image compression and video denoising, is essential.


\bibliographystyle{plain}
\bibliography{references}

\end{document}